\documentclass[12pt]{article}
\usepackage{longtable,supertabular}
\usepackage{amsmath}
\usepackage{amssymb}
\usepackage{amsthm}
\usepackage{latexsym}
\usepackage{cite}
\usepackage{a4,color}

\newcommand{\bC}{{\mathbf{{C}}}}
\newcommand{\bF}{{\mathbf{{F}}}}
\newcommand{\bT}{{\mathbf{{T}}}}
\newcommand{\bZ}{{\mathbf{{Z}}}}

\newcommand{\bu}{{\mathbf{u}}}
\newcommand{\bv}{{\mathbf{v}}}

\newcommand{\ord}{{\mathrm{ord}}}
\newcommand{\lcm}{{\mathrm{lcm}}}
\newcommand{\wt}{{\mathrm{wt}}}
\newcommand{\m}{{\mathsf{M}}}
\newcommand{\bzero}{{\mathbf{0}}}

\newcommand{\Plotkin}{{\mathrm{Plotkin}}}

\newtheorem{theorem}{Theorem}
\newtheorem{lemma}[theorem]{Lemma}

\newtheorem{problem}{Open Problem}

\newtheorem{example}{Example}

\title{Self-Dual Cyclic Codes with Square-Root-Like Lower Bounds on Their Minimum Distances\thanks{The research of Hao Chen was supported by NSFC Grant 62032009. The research of Cunsheng Ding was supported by The Hong Kong Research Grants Council under Grant No. 16301123.}}
\author{Hao Chen\thanks{H. Chen is with the College of Information Science and Technology/Cyber Security, Jinan University, Guangzhou, Guangdong Province, 510632, China (e-mail: haochen@jnu.edu.cn).}
\and Cunsheng Ding\thanks{C. Ding is with the Department of Computer Science and Engineering, The Hong Kong University of Science and Technology, Hong Kong, China (cding@ust.hk).}}

\begin{document}

\maketitle

\begin{abstract}
Binary self-dual cyclic codes have been studied since the classical work of Sloane and Thompson published in IEEE Trans. Inf. Theory, vol. 29, 1983.   Twenty five years later, an infinite family of binary self-dual cyclic codes with lengths $n_i$ and minimum distances $d_i \geq \frac{1}{2} \sqrt{n_i+2}$ was presented in a paper of IEEE Trans. Inf. Theory, vol. 55, 2009. However, no infinite family of Euclidean self-dual binary 
cyclic codes whose minimum distances have the square-root lower bound 
and no infinite family of Euclidean self-dual nonbinary 
cyclic codes whose minimum distances have a lower bound better than the square-root lower bound are known in the literature.  
In this paper, an infinite family of Euclidean 
self-dual cyclic codes over the fields ${\bf F}_{2^s}$ with a square-root-like lower bound 
is constructed. 
An infinite subfamily of this family consists of self-dual binary cyclic codes with the square-root lower bound. 
Another infinite subfamily of this family consists of self-dual cyclic codes over the fields ${\bf F}_{2^s}$ 
with a lower bound better than 
the square-root bound for $s \geq 2$.  Consequently, two breakthroughs in coding theory are made in this paper. 
An infinite family of self-dual binary cyclic codes with a square-root-like lower bound is also 
presented in this paper. 
An infinite family of Hermitian self-dual  cyclic codes over the fields ${\bf F}_{2^{2s}}$
with a square-root-like lower bound and an infinite family of Euclidean self-dual linear codes over 
${\bf F}_{q}$ with $q \equiv 1 \pmod{4}$ with a square-root-like lower bound are also constructed in this paper.  \\

{\bf Index terms:} Cyclic code, linear code, self-dual code.  
\end{abstract}

\section{Introduction}

\subsection{Basics of linear codes}

The \emph{Hamming weight}  of a vector ${\bf a}=(a_0, \ldots, a_{n-1}) \in {\bf F}_q^n$ is defined by
$$\wt({\bf a})=|\{i: a_i \neq 0\}|.$$
The \emph{Hamming distance} $d({\bf a}, {\bf b})$ between two vectors ${\bf a}$ and ${\bf b}$ is defined by
$$d({\bf a}, {\bf b})=\wt({\bf a}-{\bf b}).$$
The \emph{Hamming distance} $d({\bf C})$ of a subset ${\bf C} \subseteq {\bf F}_q$ is defied by
$$d({\bf C})=\min_{{\bf a} \neq {\bf b}} \{d({\bf a}, {\bf b}) :  {\bf a} \in {\bf C}, {\bf b} \in {\bf C} \}.$$
An $[n, k, d]_q$ \emph{linear code} is a linear subspace of ${\bf F}_q^n$ with dimension $k$ and
minimum distance $d$. Throughout this paper, by an $[n, k, d]_q$ code we mean an $[n, k, d]_q$ linear code
over $\bF_q$.
It is clear that the minimum Hamming distance of a nonzero linear code is its minimum nonzero Hamming weight.   The Singleton bound asserts that $d \leq n-k+1$ for an $[n, k, d]_q$ code. A linear code attaining this bound is said to be  maximal distance separable (MDS). Reed-Solomon codes are well-known MDS codes (see, e.g.,  \cite{HP,Lint,MScode}). \\

The \emph{Euclidean inner product} on ${\bf F}_q^n$ is defined by $$<{\bf x}, {\bf y}>=\Sigma_{i=1}^n x_iy_i,$$ where ${\bf x}=(x_1, \ldots, x_n)$ and ${\bf y}=(y_1, \ldots, y_n)$. The \emph{Euclidean dual} of a linear code ${\bf C}\subseteq {\bf F}_q^n$ is $${\bf C}^{\perp}=\{{\bf c} \in {\bf F}_q^n: <{\bf c}, {\bf y}>=0, \forall {\bf y} \in {\bf C}\}.$$ A linear code ${\bf C} \subseteq {\bf F}_q^n$ is \emph{Euclidean self-orthogonal} if ${\bf C} \subseteq {\bf C}^{\perp}$, \emph{Euclidean dual-containing} if ${\bf C}^{\perp} \subseteq {\bf C}$, and \emph{Euclidean  self-dual} if ${\bf C}={\bf C}^{\perp}$. The Euclidean dual of a Euclidean dual-containing code is a Euclidean self-orthogonal code.

Let $q=q_1^2$.
Similarly the \emph{Hermitian inner product} on ${\bf F}_{q}^n$ is defined by $$<{\bf x}, {\bf y}>_H=\Sigma_{i=1}^n x_iy_i^{q_1},$$ where ${\bf x}=(x_1, \ldots, x_n)$ and ${\bf y}=(y_1, \ldots, y_n)$ are two vectors in ${\bf F}_{q}^n$. The \emph{Hermitian dual} of a linear code ${\bf C}\subseteq {\bf F}_{q}^n$ is $${\bf C}^{\perp_H}=\{{\bf c} \in {\bf F}_{q}^n: <{\bf c}, {\bf y}>_H=0, \forall {\bf y} \in {\bf C}\}.$$ For a linear code ${\bf C} \subseteq {\bf F}_{q}^n$,
we set $${\bf C}^{q_1}=\{(c_0^{q_1}, \ldots, c_{n-1}^{q_1}): (c_0, \ldots, c_{n-1}) \in {\bf C}\}.$$
Then it is clear that
\begin{eqnarray}\label{eqn-EHrelation}
{\bf C}^{\perp_H}=({\bf C}^{\perp})^{q_1}=({\bf C}^{q_1})^{\perp}.
\end{eqnarray}
A linear code ${\bf C} \subseteq {\bf F}_{q}^n$ is \emph{Hermitian self-orthogonal} if ${\bf C} \subseteq {\bf C}^{\perp_H}$, \emph{Hermitian dual-containing} if ${\bf C}^{\perp_H} \subseteq {\bf C}$, and \emph{Hermitian self-dual} if ${\bf C}={\bf C}^{\perp_H}$. The Hermitian dual of a Hermitian dual-containing code is a Hermitian self-orthogonal code.

\subsection{The importance of cyclic codes in theory and practice}

Let ${\bf C} \subseteq {\bf F}_q^n$ be a linear code. If $(c_0, c_1, \ldots, c_{n-1}) \in {\bf C}$ implies
$(c_{n-1}, c_0,  \ldots, c_{n-2})$ $ \in {\bf C}$, then this code ${\bf C} \subseteq {\bf F}_q^n$ is said to be
\textit{cyclic}. A codeword ${\bf c}$ in a cyclic code  ${\bf C} \subseteq {\bf F}_q^n$ is identified with the polynomial ${\bf c}(x)=c_0+c_1x+\cdots+c_{n-1}x^{n-1}\in {\bf F}_q[x]/(x^n-1)$. With this identification,
every cyclic code  ${\bf C} \subseteq {\bf F}_q^n$ is a principal ideal in the ring ${\bf F}_q[x]/(x^n-1)$ and is then generated by many polynomials in ${\bf F}_q[x]/(x^n-1)$. Let ${\bf C}=(g(x))$, where $g(x)$ is the monic polynomial with the lowest degree,
then $g(x)$ must be a divisor of $x^n-1$ and is called the \emph{generator polynomial} of ${\bf C}$, and $h(x)=\frac{x^n-1}{g(x)}$ is referred to as
the \emph{check polynomial} of ${\bf C}$.
The dual code ${\bf C}^\perp$ of a cyclic code ${\bf C}$ with generator polynomial $g(x)$ is the cyclic code with the generator polynomial $h^{\perp}(x)=\frac{x^k h(x^{-1})}{h(0)}$, which is the reciprocal polynomial of the check polynomial $h(x)$ of ${\bf C}$. \\

Cyclic codes, in particular, the Bose-Chaudhuri-Hocquenghem (BCH) codes introduced in 1959-1960 (see \cite{BC1,BC2,Hoc}), are the most important type of linear codes in coding theory and practice due to the
following:
\begin{itemize}
\item They have the cyclic structure and the efficient Berlekamp-Massey decoding algorithm.
\item They (in particular, the cyclic Reed-Solomon codes) are widely used in communications and data storage  systems.
\item They are closely related to combinatorics, number theory, group theory and many other areas of mathematics.
\end{itemize}
For similar reasons, the Euclidean and Hermitian self-dual cyclic codes are more interesting than other
Euclidean and Hermitian self-dual non-cyclic linear codes. This explains why we are more interested in
Euclidean and Hermitian self-dual cyclic codes.
\\

The construction and decoding of cyclic and BCH codes have been extensively studied
in the literature (see \cite[Chapters 3, 7, 8, 9]{MScode}, \cite[Chapters 4, 5, 6]{HP},
and \cite[Chapter 6]{Lint}). The total number of references on cyclic codes in the literature is huge and this shows the importance of cyclic codes in another sense.

\subsection{The motivations and objectives of this paper}\label{sec-motivation}

A lower bound on
the minimum distance of an $[n,k,d]_q$ code is called  the \emph{square-root lower bound} if it is
$\sqrt{n}$, and
a \emph{square-root-like lower bound} if it is at least
$c\sqrt{n}$ for a fixed positive constant $c$ and fixed $q$. \\

Self-dual codes have been a very hot and interesting topic partially due to the following:
\begin{itemize}
\item Some of them support combinatorial $5$-designs. For example, the extended binary and ternary Golay codes and some Pless symmetry codes are such codes \cite{dingtang2022,HP}.
\item They are related to invariant theory \cite{Rains}.
\item They have applications in quantum error-correcting codes \cite{Rains}.
\end{itemize}
To the best knowledge of the authors, known infinite families of self-dual codes with unbounded length $n$ and minimum distance $d\geq \sqrt{n}$ are the following:
\begin{enumerate}
\item The extended codes of odd-like quadratic residue codes  \cite{HP}, which are not cyclic.
\item The Pless symmetry codes \cite{Pless1972,dingtang2022}, which are not cyclic.
\item The generalized Reed-Muller code of order $[m(q-1)-1]/2$ over ${\bf F}_q$ with parameters $[q^m, q^m/2, [(q+2)/2]q^{(m-1)/2}]$, where $m$ is odd and $q=2^s$ with a positive integer $s$ \cite[Theorem 5.8, Theorem 5.25]{AK98}, which are not cyclic. 	
\item A family of negacyclic codes over ${\bf F}_q$ with parameters (see \cite{SunDW})
$$\left[\frac{q^m-1}2, \frac{q^m-1}4, d  \right ],$$
where $q$ is an odd prime power, $m\geq 2$ is an integer, $q^m\equiv 1\pmod{4}$ and
\begin{align*}
	d=\begin{cases}
	q^{m/2}~&{\rm if}~m~{\rm is~even},\\
	(\frac{q+3}{4})q^{(m-1)/2}~&{\rm if}~m~{\rm is~odd}.
\end{cases}
\end{align*}
\end{enumerate}
Notice that all the families of self-dual codes above are not cyclic. Hence, no infinite family of self-dual cyclic
codes whose minimum distances are better than the square-root bound are known in the literature. The main  motivation of this paper is the lack of such an infinite family of self-dual cyclic codes. \\

While it is hard to construct an infinite family of self-dual cyclic codes with the square-root lower bound,
it is natural for people to construct an infinite family of self-dual cyclic codes with a square-root-like lower bound.
In 2009, an important progress in this direction was made and an infinite family of binary self-dual cyclic codes
with length $n=2(2^{2a+1}-1)$ and minimum distance $\delta \geq \frac{1}{2} \sqrt{n+2}$ were constructed
by Heijne and Top  \cite{HTop}. Constructing infinite families of self-dual cyclic codes with
a  square-root-like lower bound is not easy and is thus the second motivation of this paper.  \\

The objectives of this paper are the following:
\begin{itemize}
\item Generalize earlier results on the construction of binary cyclic codes.
\item Present a new method for constructing self-dual linear codes.
\item Construct infinite families of Euclidean and Hermitian self-dual cyclic codes with a square-root-like lower bound.
\item Construct infinite families of Euclidean self-dual linear  codes with a square-root-like lower bound.
\end{itemize}

\subsection{The organization of this paper}

The rest of this paper is arranged as follows.
Section \ref{sec-evanlint} generalises the van Lint construction of the repeated-root binary cyclic code of length $2n$ for odd $n$, and develops a generalised van Lint construction.
Section \ref{sec-constructs} presents a $[{\bf u}|{\bf u}+{\bf v}]$ construction of Euclidean and Hermitian self-dual cyclic codes over $\bF_{2^s}$ via the generalised
van Lint theorem.
Section \ref{sec-BCHcontain} recalls two families of dual-containing BCH codes.
Section \ref{sec-code1}
constructs families of Euclidean and Hermitian self-dual cyclic codes with a
square-root-like lower bound on their minimum distances.
Section \ref{sec-code2} constructs a family of Euclidean self-dual linear codes over
$\bF_q$ for $q \equiv 1 \pmod{4}$ with a
square-root-like lower bound on their minimum distances.
Section \ref{sec-Oct16} presents a family of self-dual binary cyclic codes with a square-root-like lower bound.
Section \ref{sec-othercodes} provides information on some self-dual linear codes with a square-root-like lower bound.
Section \ref{sec-final} summarises the contributions of this paper and makes some concluding remarks.

\section{A generalized  van Lint theorem}\label{sec-evanlint}

Let $\bC_1$ and $\bC_2$ be two linear codes with parameters $[n,k_1, d_1]_q$ and $[n,k_2, d_2]_q$, respectively.
The Plotkin sum of $\bC_1$ and $\bC_2$ is denoted by $\Plotkin(\bC_1, \bC_2)$ and defined by
\begin{eqnarray*}
\Plotkin(\bC_1, \bC_2):=\{(\bu|\bu+\bv):  \bu \in \bC_1, \, \bv \in \bC_2 \},
\end{eqnarray*}
where $\bu | \bv$ denotes the concatenation of the two vectors $\bu$ and $\bv$. It is well known that
$\bC$ has parameters $[2n, k_1+k_2, \min\{2d_1, d_2\}]_q$ (see \cite{MScode}). The Plotkin sum is also called
the  $[{\bf u}|{\bf u}+{\bf v}]$ construction and was introduced in 1960 by Plotkin \cite{Plotkin}. \\

The $[{\bf u}|{\bf u}+{\bf v}]$ construction of binary repeated-root cyclic codes was given by van Lint in his paper \cite{Lint1}.  By modifying the conditions of Theorem 1 in \cite{Lint1}, we have the following theorem, which is a  generalization of the original van Lint theorem.

\begin{theorem}[The generalized van Lint theorem]\label{thm-evanlint}
Let $q$ be a power of $2$ and $n$ be an odd positive integer. Let ${\bf C}_1 \subseteq {\bf F}_q^n$ be a cyclic code
with generator polynomial $g_1(x) \in {\bf F}_q[x]$ and ${\bf C}_2 \subseteq {\bf F}_q^n$ be a cyclic code
generated by the polynomial $g_1(x)g_2(x) \in {\bf F}_q[x]$, where $g_2(x)$ is a divisor of $x^n+1$.
Then the code $\Plotkin(\bC_1, \bC_2)$
is permutation-equivalent to  the repeated-root cyclic code of length $2n$ generated by the polynomial $g_1(x)^2g_2(x)$.
\end{theorem}

\begin{proof}
Let $\bC$ denote $\Plotkin(\bC_1, \bC_2)$.
Let $f(x)=\gcd(g_1(x), g_2(x))$. Then $f(x)$ is a divisor of $x^n+1$.
By definition, $\bC_2$ has generator polynomial
$$\gcd(x^n+1, g_1(x)g_2(x))=g_1(x)g_2(x)/f(x)$$
and thus dimension
$$
n-\deg(g_1(x)) -\deg(g_2(x))+\deg(f(x)).
$$
It then follows that
\begin{eqnarray*}
\dim(\bC) &=& \dim(\bC_1) + \dim(\bC_2) \nonumber \\
                 &=& 2n-2\deg(g_1(x)) -\deg(g_2(x))+\deg(f(x)).
\end{eqnarray*}

Let $\bC'$ denote the the repeated-root cyclic code of length $2n$ generated by the polynomial $g_1(x)^2g_2(x)$. By definition, $\bC'$ has generator polynomial
$$
\gcd(x^{2n}+1,g_1(x)^2g_2(x))= \frac{g_1(x)^2g_2(x)}{f(x)},
$$
as $x^n+1$ has no repeated roots due to the assumption that $n$ is odd.
Consequently, $\bC'$ has dimension
$$
2n-2\deg(g_1(x)) -\deg(g_2(x))+\deg(f(x)).
$$
We have then
\begin{eqnarray}\label{eq-maindim}
\dim(\bC)=\dim(\bC').
\end{eqnarray}

Let ${\bf a}=(a_0, \ldots, a_{n-1}) \in {\bf C}_1$ and ${\bf c}=(c_0, \ldots, c_{n-1}) \in {\bf C}_2$. Set ${\bf b}={\bf a}+{\bf c}=(b_0, \ldots, b_{n-1})$. Then ${\bf a}(x)=a_0+a_1x+\cdots+a_{n-1}x^{n-1}={\bf a}_{even}(x^2)+x{\bf a}_{odd}(x^2)$ is divisible by $g_1(x)$. And ${\bf c}(x)=c_0+c_1x+\cdots+c_{n-1}x^{n-1}$ is divisible by $\gcd(x^n+1, g_1(x)g_2(x))$. Therefore ${\bf b}(x)=b_0+b_1x+\cdots+b_{n-1}x^{n-1}={\bf b}_{even}(x^2)+x{\bf b}_{odd}(x^2)$ is divisible by $g_1(x)$.\\

Set $${\bf w}=(a_0, b_1, a_2, b_3, \ldots,b_{n-2}, a_{n-1}, b_0, a_1, \ldots, a_{n-2}, b_{n-1}).$$ Then
$${\bf w}(x)=[{\bf a}_{even}(x^2)+x^{n+1}{\bf a}_{odd}(x^2)]+[x{\bf b}_{odd}(x^2)+x^n{\bf b}_{even}(x^2)].$$ It is clear that $${\bf a}_{even}(x^2)+x^{n+1}{\bf a}_{odd}(x^2)={\bf a}(x)+x(x^n+1){\bf a}_{odd}(x^2)$$ is divisible by $g_1(x)$, since $g_1(x)$ is a divisor of $x^n+1$. Since there is only even degree powers $x^{2h}$'s in ${\bf a}(x)+x(x^n+1){\bf a}_{odd}(x^2)$, this term is divisible by $g_1(x)^2$, from the fact that $g_1(x)$ has no repeated root. The second term $x{\bf b}_{odd}(x^2)+x^n{\bf b}_{even}(x^2)={\bf b}(x)+(x^n+1){\bf b}_{odd}(x^2)$ is divisible by $g_1(x)^2$, from a similar argument.\\

On the other hand, ${\bf w}(x)=(x^n+1){\bf a}(x)+{\bf c}(x)+(x^n+1){\bf c}_{even}(x^2)$ is divisible by $g_2(x)$, then ${\bf w}(x)$ is divisible by $\gcd(x^{2n}+1, g_1(x)^2g_2(x))$. Then the code ${\bf C}$ is permutation-equivalent to the cyclic code $\bC'$ of length $2n$ generated by $g_1(x)^2g_2(x)$, as both codes have the same dimension by \eqref{eq-maindim}. The desired conclusion then follows.
\end{proof}

We have the following remarks regarding Theorem \ref{thm-evanlint}:
\begin{itemize}
\item By definition, $\bC_2$ is a subcode of $\bC_1$ and it is possible that $\bC_2=\bC_1$.
\item The proof of Theorem \ref{thm-evanlint} shows that the cyclic code $\bC'$ has generator polynomial
$$
\gcd(x^{2n}+1,g_1(x)^2g_2(x))= \frac{g_1(x)^2g_2(x)}{f(x)}
$$
 and the codes $\Plotkin(\bC_1, \bC_2)$ and $\bC'$ have dimension
$$2n-2\deg(g_1(x)) -\deg(g_2(x))+\deg(f(x)).$$
\item The code $\Plotkin(\bC_1, \bC_2)$ is Euclidean (Hermitian) self-dual if and only if the cyclic code
$\bC'$ is  Euclidean (Hermitian) self-dual, as the two codes are permutation-equivalent.
\item In the original theorem proved by van Lint (see \cite[Section 6.10]{Lint}), the condition that $\gcd(g_1(x), g_2(x))=1$ is required. Theorem \ref{thm-evanlint} does not have such a condition. Hence,
Theorem \ref{thm-evanlint}  is indeed a generalization of the original van Lint theorem. Even in the case
$q=2$, Theorem \ref{thm-evanlint} extends  the original van Lint theorem to a large extent.
\end{itemize}

\section{The $[{\bf u}|{\bf u}+{\bf v}]$ construction of Euclidean and Hermitian self-dual cyclic codes over $\bF_{2^s}$}\label{sec-constructs}

In this section, we give the main construction of self-dual codes over ${\bf F}_{2^s}$ from dual-containing codes.\\

\begin{theorem}\label{thm-3.1}
Let ${\bf F}_{2^s}$ be the finite field with $2^s$ elements. Suppose that ${\bf C} \subseteq {\bf F}_{2^s}^n$ is a  Euclidean dual-containing linear code. Then the linear code $${\bf C}_1=\{{\bf u}|{\bf u}+{\bf v}: {\bf u} \in {\bf C}, {\bf v} \in {\bf C}^{\perp}\} \subseteq {\bf F}_{2^s}^{2n}$$ is a self-dual code with minimum distance $\min\{d({\bf C}^{\perp}), 2d({\bf C})\}$.

Let ${\bf F}_{2^{2s}}$ be the finite field with $2^{2s}$ elements. Suppose that ${\bf C} \subseteq {\bf F}_{2^{2s}}^n$ is a  Hermitian dual-containing linear code. Then the linear code $${\bf C}_2=\{{\bf u}|{\bf u}+{\bf v}: {\bf u} \in {\bf C}, {\bf v} \in {\bf C}^{\perp_H}\} \subseteq {\bf F}_{2^{2s}}^{2n}$$ is a Hermitian self-dual code with minimum distance
$$
\min\{d({\bf C}^{\perp_H}), 2d({\bf C})\}=\min\{d({\bf C}^{\perp}), 2d({\bf C})\}.
$$
\end{theorem}

\begin{proof}
For two codewords ${\bf c}_1=[{\bf u}_1|{\bf u}_1+{\bf v}_1]$ and ${\bf c}_2=[{\bf u}_2|{\bf u}_2+{\bf v}_2]$ in this code ${\bf C}_1$, their Euclidean inner product is $$<{\bf c}_1,{\bf c}_2>=(1+1)<{\bf u}_1, {\bf u}_2>+<{\bf u}_1, {\bf v}_2>+<{\bf u}_2,{\bf v}_1>+<{\bf v}_1, {\bf v}_2>.$$ Since ${\bf C}^{\perp}$ is self-orthogonal, this inner product is zero. It is clear that the dimension of ${\bf C}_1$ is $$\dim({\bf C}_1)=\dim({\bf C}^{\perp})+\dim({\bf C})=n.$$ It is well known
that
$$
d(\bC_1)=d(\Plotkin(\bC, \bC^\perp)) = \min\{ 2d(\bC), d(\bC^\perp) \}.
$$
Then the first conclusion follows immediately. The second conclusion can be similarly proved with the well known fact that
$d({\bf C}^{\perp_H})=d({\bf C}^{\perp})$ (it follows from \eqref{eqn-EHrelation}).
\end{proof}

Suppose that the length of the cyclic code ${\bf C}$ over ${\bf F}_{2^s}$ is odd and the generator polynomial of ${\bf C}$ is $g_1(x) \in {\bf F}_{2^s}[x]$. Since ${\bf C}^{\perp} \subseteq {\bf C}$,  the generator polynomial of ${\bf C}^{\perp}$ is $g_1(x)g_2(x) \in {\bf F}_{2^s}[x]$, where $g_2(x)$ is a divisor of $x^n+1$. By the generalized van Lint Theorem (see Theorem \ref{thm-evanlint}), the code $\bC_1$  is permutation-equivalent to the repeated-root cyclic code $\bC'$ over ${\bf F}_{2^s}$ with length $2n$ and generator polynomial $g_1^2(x) g_2(x)$. Therefore,
if the building block ${\bf C}$ in Theorem \ref{thm-3.1}  is a cyclic code with odd length $n$, then the code ${\bf C}_1$ is
permutation-equivalent to a repeated-root self-dual cyclic code of the length $2n$. In the subsequent sections,
we will construct specific families of dual-containing cyclic codes $\bC$ and then plug them into the construction
of Theorem \ref{thm-3.1} to obtain
families of self-dual cyclic codes with very good minimum distances.

\section{Euclidean and Hermitian dual-containing BCH codes}\label{sec-BCHcontain}

\subsection{A useful lemma}

The following lemma will be needed later.

\begin{lemma}\label{lem-4.1}
(see \cite{LDL}) Let $q>1$ be a positive integer and $a,b$ be two positive integers.  Then $\gcd(q^a-1, q^b-1)=q^{\gcd(a,b)}-1$, $\gcd(q^a+1, q^b-1)=1$ if $q$ is even and $\frac{b}{\gcd(a,b)}$ is odd,
$\gcd(q^a+1, q^b-1)=2$ if $q$ is odd and $\frac{b}{\gcd(a,b)}$ is odd, and $\gcd(q^a+1, q^b-1)=q^{\gcd(a,b)}+1$ if $\frac{b}{\gcd(a,b)}$ is even.
\end{lemma}

\subsection{BCH codes}\label{sec-BCHcode}

In this subsection, we recall the basics of BCH codes.
Let $\gcd(q,n)=1.$ Define ${\bf Z}_n={\bf Z}/n{\bf Z}=\{0, 1, \ldots, n-1\}$. For each $i \in \bZ_n$,
the \emph{$q$-cyclotomic coset} $C_i$ containing $i$ is defined by
$$C_i=\{ iq^{i} \bmod{n} : 0 \leq i \leq l-1\},$$
where $l$ is the smallest positive integer such that $iq^l \equiv i \pmod{n}$ and
$a \bmod{n}$ denotes the unique $b \in \bZ_n$ such that $a \equiv b \pmod{n}$. The smallest integer
in $C_i$ is called the \emph{coset leader} of $C_i$. Let $\Gamma(q,n)$ denote the set of all coset leaders. Then
$\{C_i: i \in \Gamma(q,n)\}$ is a partition of $\bZ_n$. \\

 Let $m=\ord_n(q)$, which is the smallest positive integer $\ell$ such that $q^\ell \equiv 1 \pmod{n}$.
 Let $\alpha$ be a primitive element of $\bF_{q^m}$. Define $\beta=\alpha^{(q^m-1)/n}$.
 Then $\beta$ is an $n$-th primitive root of unity. Define
 $$
 \m_{\beta^i}(x)=\prod_{j \in C_i} (x-\beta^j).
 $$
 It is easily seen that $ \m_{\beta^i}(x)$ is an irreducible polynomial in $\bF_q[x]$ and the canonical factorization
 of $x^n-1$ over $\bF_q$ is
 $$
 x^n-1=\prod_{j \in \Gamma(q,n)}  \m_{\beta^i}(x).
 $$
 By definition, the generator polynomial $g(x)$ of a cyclic code $\bC$ of length $n$ over $\bF_q$ is a divisor
 of $x^n-1$. The \emph{defining set} of the cyclic code $\bC$ with respect to $\beta$ is defined by
\begin{eqnarray}\label{eqn-definingset}
{\bf T}_{{ g}}=\{i \in \bZ_n: {g}(\beta^i)=0\}.
\end{eqnarray}
Then the defining set of a cyclic code is the disjoint union of some $q$-cyclotomic cosets. \\

Let $\delta \geq 2$ be an integer and $b$ be an integer. Define
\begin{eqnarray}\label{eqn-gpBCHcode}
g_{(q,n,\delta,b, \beta)}(x)=\lcm\{ \m_{\beta^b}(x), \ldots, \m_{\beta^{b+\delta-2}}(x)\},
\end{eqnarray}
where $\lcm$ denotes the least common multiple of the set of polynomials over $\bF_q$.
Let $\bC_{(q,n,\delta,b, \beta)}$ denote the cyclic code over $\bF_q$ with length $n$ and
generator polynomial $g_{(q,n,\delta,b, \beta)}(x)$. This code  $\bC_{(q,n,\delta,b, \beta)}$
is called a \emph{BCH code} with designed distance $\delta$. It is well known that the
minimum distance of $\bC_{(q,n,\delta,b, \beta)}$ is lower bounded by $\delta$.
This follows from the famous BCH bound for cyclic codes \cite{HP,MScode}.

\subsection{Some Euclidean and Hermitian dual-containing BCH codes}

The purpose of this section is to introduce a family of known Euclidean dual-containing BCH codes and a family
of known Hermitian dual-containing BCH codes. These known codes will be used as basic building blocks in the constructions of self-dual codes in later sections.

Throughout this section, let $q$ be a prime power and $n$ be a positive integer with $\gcd(q,n)=1$.
Define ${\bf Z}_n={\bf Z}/n{\bf Z}=\{0, 1, \ldots, n-1\}$. For any subset $\bT \subseteq \bZ_n$, define
$$
\bT^{-1}=\{ n-i: i \in \bT\}
$$
and
$$
\bT^c=\bZ_n \setminus \bT,
$$
which is the complement set of $\bT$. In the case that $q=q_1^2$, define
$$
\bT^{-q_1}=\{(n-q_1i) \bmod{n}: i \in \bT\}.
$$

%\subsubsection{A general construction of Euclidean dual-containing BCH codes}

We follow the notation of Section \ref{sec-BCHcode}. Let $\bC$ be a cyclic code over $\bF_q$ with length $n$
and generator polynomial $g(x)$, and $\beta$ be an $n$-th primitive root of unity in $\bF_{q^m}$. Let
$\bT_g$ be the defining set of $\bC$ with respect to $\beta$ defined in \eqref{eqn-definingset}.
Then the defining set of the Euclidean dual code $\bC^\perp$ is $({\bf T}_g^c)^{-1}$,
where ${\bf T}_g^c={\bf Z}_n \setminus {\bf T}_g$ is the complementary set of $\bT_g$.
The following result is well known  \cite{LDL}.

\begin{lemma}\label{lem-oct11Eucl}
The following hold:
\begin{itemize}
\item The code $\bC$ is Euclidean dual-containing if and only if ${\bf T}_g \cap {\bf T}_g^{-1}=\emptyset$.
\item The code $\bC$ is Hermitian dual-containing if and only if ${\bf T}_g \cap {\bf T}_g^{-q_1}=\emptyset$,
\end{itemize}
where $q=q_1^2$ and
$$
\bT_g^{-q_1}=\{(n-q_1i) \bmod{n}: i \in \bT_g\}.
$$
\end{lemma}

The BCH lower bound on the minimum distances of cyclic codes can be used to construct Euclidean and Hermitian dual-containing codes with a square-root-like lower bound on their minimum distances. Then such codes will be
used to construct self-dual cyclic codes with a square-root-like lower bound on their minimum distances later. \\

The following result is a special case of Theorem 3 in \cite{Aly} and will be needed later.

\begin{theorem}\label{thm-4.2}
Let $m$ be odd.
Let $\mu $ be a proper divisor of $q^m-1$ and $n=\frac{q^m-1}{\mu} >4$.
Put
$$
\delta =\left\lfloor \frac{q^{\frac{m+1}{2}}-q+1}{\mu} \right\rfloor.
$$
Then $\bC_{(q,n,\delta,1, \beta)}$ is a Euclidean dual-containing BCH code with minimum distance at least $\delta$.
\end{theorem}

\begin{proof}
It follows from Lemma \ref{lem-4.1} that $\ord_n(q)=m$. The desired conclusion then follows from
\cite[Theorem 3]{Aly}.
\end{proof}

%\begin{proof}
%By Lemma \ref{lem-oct11Eucl}, we only need to prove that there are no two positive integer $1 \leq u, v \leq  \frac{q^{\frac{m+1}{2}}-q}{\mu}$ such that $uq^t+v \equiv 0 \pmod{n}$, where $t < m$. Without the loss of the generality, we can assume that $t \leq \frac{m-1}{2}$. Otherwise $t \geq \frac{m+1}{2}$, $uq^m+vq^{m-t}\equiv u+vq^{m-t} \equiv 0 \pmod{n}$. Then $\mu(uq^t+v) \geq q^m-1$,which is impossible. Them the desired conclusion follows.
%\end{proof}

For Hermitian dual-containing BCH codes, we have the following result. The proof is similar to the proof of
Theorem \ref{thm-4.2} with the help of Lemma \ref{lem-oct11Eucl} (see \cite{Aly}).

\begin{theorem}\label{thm-4.3}
Suppose that $q=q_1^2$.
Let $\mu $ be a proper divisor of $q_1^{m}-1$, $m$ be odd, and $n=\frac{q^m-1}{\mu}$. Put
$\delta=  (q_1^m-1)/\mu $.
Then $\bC_{(q,n,\delta,1, \beta)}$ is a Hermitian dual-containing BCH code with minimum distance at least $\delta$.
\end{theorem}

\section{Self-dual cyclic codes over ${\bf F}_{2^s}$ with a square-root-like lower bound}\label{sec-code1}

Combining Theorems \ref{thm-evanlint}, \ref{thm-3.1}, \ref{thm-4.2} and \ref{thm-4.3},
we will construct  Euclidean and Hermitian self-dual cyclic codes over ${\bf F}_q$ for even $q$  with a
square-root-like lower bound on their minimum distances.

\begin{theorem}\label{thm-5.1}
Let $q=2^s$ with $s \geq 1$ and $m$ be odd.
Let $\mu $ be a proper divisor of $q^m-1$ and $n=\frac{q^m-1}{\mu} >4$.
Put
$$
\delta =\left\lfloor \frac{q^{\frac{m+1}{2}}-q+1}{\mu} \right\rfloor.
$$
Let $h_{(q,n,\delta,1, \beta)}^\perp(x)$ denote the reciprocal polynomial of $(x^n+1)/g_{(q,n,\delta,1, \beta)}(x)$,
where $g_{(q,n,\delta,1, \beta)}(x)$ is the generator polynomial of $\bC_{(q,n,\delta,1, \beta)}$.
Let $\bC'$ denote the cyclic code of length $2n$ over $\bF_q$ with generator polynomial
$g_{(q,n,\delta,1, \beta)}(x) h_{(q,n,\delta,1, \beta)}^\perp(x)$. If $\delta \geq 2$,
then
$\bC'$ is a Euclidean self-dual cyclic code with parameters
$
[2n, n, d],
$
where
$$
d=\min\{2d(\bC_{(q,n,\delta,1, \beta)}), \, d(\bC_{(q,n,\delta,1, \beta)}^\perp)\} \geq \delta.
$$
\end{theorem}

\begin{proof}
Suppose that $\delta \geq 2$.
By definition, $h_{(q,n,\delta,1, \beta)}^\perp(x)$ is the generator polynomial of $\bC_{(q,n,\delta,1, \beta)}^\perp$.
Let
$$\bC=\Plotkin( \bC_{(q,n,\delta,1, \beta)}, \bC_{(q,n,\delta,1, \beta)}^\perp).
$$
By Theorem \ref{thm-evanlint}, $\bC$ is permutation-equivalent to $\bC'$.
The desired conclusions then follow from Theorem \ref{thm-3.1} and Theorem \ref{thm-4.2}.
\end{proof}

We have the following remarks regarding the cyclic code $\bC'$ in Theorem \ref{thm-5.1}.
\begin{itemize}
\item The lower bound $d \geq \delta$ on the minimum distance of the cyclic code $\bC'$ in Theorem \ref{thm-5.1}
is a square-root-like lower bound.
\item When $\mu=1$, the cyclic code $\bC'$ in Theorem \ref{thm-5.1} has parameters
$$
[2(q^m-1), q^m-1, d \geq q^{(m+1)/2}-q+1].
$$
In this case, the lower bound $d \geq q^{(m+1)/2}-q+1$ is better than the square-root bound for $q \geq 4$
\item
When $\mu=1$ and $q=2$, the cyclic code $\bC'$ in Theorem \ref{thm-5.1} has minimum distance
$d \geq 2^{(m+1)/2}-1$. Since $\bC'$ is a self-dual binary cyclic code, the Hamming weight of every
codeword of $\bC'$ must be even. It then follows from $d \geq 2^{(m+1)/2}-1$ that
\begin{eqnarray}\label{eqn-breakthrough}
d(\bC') \geq 2^{(m+1)/2} =\left \lceil \sqrt{2(2^m-1)} \right\rceil.
\end{eqnarray}
Hence, in this case  $\bC'$ is a self-dual binary cyclic code with
parameters
$$
[2(2^m-1), 2^m-1, d \geq 2^{(m+1)/2}].
$$
In this case, the lower bound $d \geq 2^{(m+1)/2}$ is exactly the square-root bound.
\end{itemize}

\begin{example}
Let $(q, \mu)=(2,1)$.
\begin{itemize}
\item When $m=3$, the cyclic code $\bC'$ in Theorem \ref{thm-5.1}  has parameters $[14,7,4]_2$.
This code has the best parameters known according to \cite{Gaborit2,Gaborit}.
\item When $m=5$, the cyclic code $\bC'$ in Theorem \ref{thm-5.1}  has parameters $[62,31,8]_2$.
\end{itemize}
\end{example}

\begin{example}
Let $(q, m, \mu)=(4,3, 3)$. Then  the cyclic code $\bC'$ in Theorem \ref{thm-5.1}  has parameters $[42, 21 , 8 ]_4$.
\end{example}

For Hermitian self-dual cyclic codes, we have the following result.

\begin{theorem}\label{thm-5.2}
Let $q_1=2^{s}$ with $s \geq 1$, $q=q_1^2$ and $m$ be odd.
Let $\mu $ be a proper divisor of $q_1^m-1$ and $n=\frac{q^m-1}{\mu} >4$.
Put
$$
\delta = \frac{q_1^{m}-1}{\mu}.
$$
Let $h_{(q,n,\delta,1, \beta)}^{\perp_H}(x)$ denote the generator polynomial of the Hermitian dual
code  $\bC_{(q,n,\delta,1, \beta)}^{\perp_H}$.
Let $\bC'$ denote the cyclic code of length $2n$ over $\bF_q$ with generator polynomial
$g_{(q,n,\delta,1, \beta)}(x) h_{(q,n,\delta,1, \beta)}^{\perp_H}(x)$,
where $g_{(q,n,\delta,1, \beta)}(x)$ is the generator polynomial of $\bC_{(q,n,\delta,1, \beta)}$.
Then
$\bC'$ is a Hermitian self-dual cyclic code with parameters
$
[2n, n, d],
$
where
$$
d=\min\{2d(\bC_{(q,n,\delta,1, \beta)}), \, d(\bC_{(q,n,\delta,1, \beta)}^\perp)\} \geq \delta.
$$

\end{theorem}

\begin{proof}
Notice that
$$
d(\bC_{(q,n,\delta,1, \beta)}^\perp)=d(\bC_{(q,n,\delta,1, \beta)}^{\perp_H}).
$$
By definition, $\delta \geq 2$.
Let
$$\bC=\Plotkin( \bC_{(q,n,\delta,1, \beta)}, \bC_{(q,n,\delta,1, \beta)}^{\perp_H}).
$$
By Theorem \ref{thm-evanlint}, $\bC$ is permutation-equivalent to $\bC'$.
The desired conclusions then follow from Theorem \ref{thm-3.1} and Theorem \ref{thm-4.3}.
\end{proof}

We have the following remarks regarding the cyclic code $\bC'$ in Theorem \ref{thm-5.2}.
\begin{itemize}
\item The lower bound $d \geq \delta$ on the minimum distance of the cyclic code $\bC'$ in Theorem \ref{thm-5.2}
is a square-root-like lower bound.
\item When $\mu=1$, the cyclic code $\bC'$ in Theorem \ref{thm-5.2} has parameters
$$
[2(q_1^{2m}-1), q_1^m-1, d \geq q_1^{m}-1].
$$
In this case, the lower bound $d \geq q_1^{m}-1$ is not as good as the square-root bound for $q \geq 4$. But
the actual minimum distance of $\bC'$ could be better than the square-root lower bound. The code example below
can justify this statement.
\end{itemize}

\begin{example}
Let $(q, m, \mu)=(4,3, 1)$. Then $\delta=7$ and the lower bound on the  minimum distance of the cyclic code $\bC'$ in Theorem \ref{thm-5.2} is $7$ only.
In this case, $ \bC_{(q,n,\delta,1, \beta)}$ has parameters $[63, 48, 7]_4$ and
$ \bC_{(q,n,\delta,1, \beta)}^{\perp}$ has parameters $[63, 15, 24]_4$. It then follows from Theorem \ref{thm-5.2}
that the cyclic code $\bC'$ in Theorem \ref{thm-5.2} in this case has parameters
$[126, 63, 14]_4$. The actual minimum distance of $\bC'$ is thus larger than the square-root lower bound.
\end{example}

\section{Self-dual linear codes over ${\bf F}_q$ with a square-root-like lower bound}\label{sec-code2}

Let $q$ be an odd prime power satisfying $q \equiv 1 \pmod{4}$ throughout this section.
Then in the field ${\bf F}_q$, $-1$ is a square.
Below we present the following variant of the construction  of Theorem  \ref{thm-3.1}. \\

\begin{theorem}\label{thm-6.1}
 Let $\lambda \in {\bf F}_q$ with $\lambda^2 =-1$. Let ${\bf C} \subseteq {\bf F}_q^n$ be a dual-containing
 linear code. Then the set
 $${\bf C}_\lambda(\bC)=\{{\bf u}|\lambda{\bf u}+{\bf v}: {\bf u} \in {\bf C}, {\bf v} \in {\bf C}^{\perp}\} \subseteq {\bf F}_q^{2n}$$
 is  a $[2n,n]_q$ self-dual linear code with minimum distance $\min\{d({\bf C}^{\perp}), 2d({\bf C})\}$.
\end{theorem}

\begin{proof}
For two codewords ${\bf c}_1=[{\bf u}_1|\lambda{\bf u}_1+{\bf v}_1]$ and ${\bf c}_2=[{\bf u}_2|\lambda{\bf u}_2+{\bf v}_2]$ in this code ${\bf C}_\lambda(\bC)$, their inner product is
$$<{\bf c}_1,{\bf c}_2>=(\lambda^2+1)<{\bf u}_1, {\bf u}_2>+\lambda(<{\bf u}_1, {\bf v}_2>+<{\bf u}_2,{\bf v}_1>)+<{\bf v}_1, {\bf v}_2>.$$
Since $\lambda^2+1=0$ and ${\bf C}^{\perp}$ is self-orthogonal, this inner product is zero. It is clear that the dimension of ${\bf C}_\lambda(\bC)$ is
$$\dim({\bf C}_\lambda(\bC))=\dim({\bf C}^{\perp})+\dim({\bf C})=n.$$
Let ${\bf u} \in {\bf C}$ and  ${\bf v} \in {\bf C}^{\perp}$. It is clear
that $({\bf u}|\lambda{\bf u}+{\bf v})$ is the zero vector if and only both $\bu$ and $\bv$ are the zero vector.
Assume that at least one of $\bu$ and $\bv$ is a nonzero vector. We have then the following cases:

\subsubsection*{The case that $\bu=\bzero$ and $\bv \neq \bzero$:}

In this case, one has
$$
\wt({\bf u}|\lambda{\bf u}+{\bf v})=\wt(\bv) \geq d(\bC^\perp).
$$

\subsubsection*{The case that  $\bu\neq \bzero$ and $\bv = \bzero$:}

In this case, one has
$$
\wt({\bf u}|\lambda{\bf u}+{\bf v})= 2\wt(\bu)\geq 2d(\bC).
$$

\subsubsection*{The case that  $\bu \neq \bzero$ and $\bv \neq \bzero$:}

In this case, one has
\begin{eqnarray*}
\wt({\bf u}|\lambda{\bf u}+{\bf v})
&=& \wt(\bu) + \wt(\lambda \bu +\bv) \\
&\geq& 2d(\bC).
\end{eqnarray*}

Clearly, there are two codewords in $\bC_\lambda(\bC)$ with Hamming weight $2d(\bC)$ and $d(\bC^\perp)$.
The desired conclusion on the minimum distance of $\bC_\lambda(\bC)$ then follows.
\end{proof}

\begin{theorem}\label{thm-6.2}
Let $m$ be odd and $q \equiv 1 \pmod{4}$.
Let $\mu $ be a proper divisor of $q^m-1$ and $n=\frac{q^m-1}{\mu} >4$.
Put
$$
\delta =\left\lfloor \frac{q^{\frac{m+1}{2}}-q+1}{\mu} \right\rfloor.
$$
Then the code $\bC_\lambda(\bC_{(q,n,\delta,1, \beta)})$ is a $[2n, n]_q$ self-dual linear code with
$$
d(\bC_\lambda(\bC_{(q,n,\delta,1, \beta)}) )= \min\{d({\bf C}^{\perp}), 2d({\bf C})\} \geq \delta
$$
\end{theorem}

\begin{proof}
By Theorem \ref{thm-4.2},  $\bC_{(q,n,\delta,1, \beta)}$ is a Euclidean dual-containing BCH code with minimum distance at least $\delta$. The desired conclusions then follow from Theorem \ref{thm-6.1}.
\end{proof}

We have the following remarks regarding the linear code $\bC_\lambda(\bC_{(q,n,\delta,1, \beta)})$ in Theorem \ref{thm-6.2}.
\begin{itemize}
\item The lower bound $d \geq \delta$ on the minimum distance of
 the code  $\bC_\lambda(\bC_{(q,n,\delta,1, \beta)})$ in Theorem \ref{thm-6.2}
is a square-root-like lower bound.
\item When $\mu=1$, the linear code $\bC_\lambda(\bC_{(q,n,\delta,1, \beta)})$ in Theorem \ref{thm-6.2} has parameters
$$
[2(q^{m}-1), q^m-1, d \geq q^{(m+1)/2}-q+1].
$$
In this case, the lower bound $ d \geq q^{(m+1)/2}-q+1$ is better than the square-root bound as $q \geq 5$ by
definition.
\end{itemize}

\section{A family of self-dual binary cyclic codes with a square-root-like lower bound}\label{sec-Oct16}

In this section, we present an infinite family of self-dual binary codes with a square-root-like lower bound.
We first prove the following theorem.

\begin{theorem}\label{thm-4.1}
Let $\gcd(q,n)=1$.
Let $\bC$ be a cyclic code of length $n$ over $\bF_q$ with generator polynomial
$g(x)$. Let $\bT_g$ be the defining set of $\bC$ with respect to an $n$-th primitive root of unity in
$\bF_{q^m}$, where $m=\ord_n(q)$. If there is no integer $a$ in ${\bf Z}_n$ such that both $a$ and $n-a$ are in the same cyclotomic coset and
 there is no integer $a \in {\bf Z}_n$ such that both $a$ and $n-a$ are in the set ${\bf T}_g$,
 then  $\bC$ is a Euclidean dual-containing cyclic code.
\end{theorem}

\begin{proof}
Suppose that there is no integer $a$ in ${\bf Z}_n$ such that
both $a$ and $n-a$ are in the same cyclotomic coset and there is no integer $a \in {\bf Z}_n$ such that
 both $a$ and $n-a$ are in the set ${\bf T}_g$.
By definition, ${\bf T}_g$ is the defining set of
$\bC$. Recall that $\bC^\perp$ has defining set $(\bZ_n\setminus  \bT_g)^{-1}$.
We would prove that
\begin{eqnarray}\label{eqn-inclusion1}
\bT_q \subset (\bZ_n\setminus  \bT_g)^{-1}.
\end{eqnarray}
Let $a \in \bT_g$. Since there is no integer $a \in {\bf Z}_n$ such that
 both $a$ and $n-a$ are in $\bT_g$, we have $n-a \in \bZ_n \setminus \bT_g$.
 Since there is no integer $a$ in ${\bf Z}_n$ such that
both $a$ and $n-a$ are in the same cyclotomic coset,
it follows from  $n-a \in \bZ_n \setminus \bT_g$ that
$a \in (\bZ_n\setminus  \bT_g)^{-1}$. Hence, the inclusion in \eqref{eqn-inclusion1} holds.
Consequently, $g(x)$ is a divisor of the generator polynomial of
$\bC^\perp$. Hence,
$
\bC^\perp \subseteq \bC.
$
\end{proof}

We are now ready to use Theorem \ref{thm-4.1} to prove the following result.

\begin{theorem}\label{thm-Oct16}
Let $n \equiv 7 \pmod{8}$ be a prime. Let $g(x)$ be the generator polynomial of an odd-like binary
quadratic residue code $\bC$ of length $n$ and let  $h^{\perp}(x)$ denote the generator polynomial
of the dual code $\bC^\perp$.
Let $\bC'$ denote the binary cyclic code of length $2n$ with generator polynomial $g(x)h^{\perp}(x)$.
Then $\bC'$ is self-dual and has parameters
$[2n, n, d \geq \lceil\sqrt{n}\, \rceil ]$.
\end{theorem}

\begin{proof}
By definition, the defining set $\bT$ of $\bC$ is either set of quadratic residues or the set of quadratic nonresidues.
Since $n \equiv 7 \pmod{8}$ is a prime, $2$ is a quadratic residue modulo $n$.
All the elements in the $2$-cyclotomic coset $C_1$ are quadratic residues.
Hence all the elements in each $2$-cyclotomic coset $C_i$ are either all quadratic residues or all
quadratic nonresidues. Notice that $-1$ is a quadric nonresidue in $\bF_n$, as $n \equiv 7 \pmod{8}$.
We then deduce that
there is no integer $a$ in ${\bf Z}_n$ such that
both $a$ and $n-a$ are in the same cyclotomic coset and there is no integer $a \in {\bf Z}_n$ such that
 both $a$ and $n-a$ are in the set ${\bf T}$. It then follows from Theorem \ref{thm-4.1} that $\bC^\perp \subset \bC$.

It is well known that $d(\bC) \geq  \lceil\sqrt{n}\, \rceil.$ It then follows from  $\bC^\perp \subset \bC$ that
$$
d(\bC^\perp) \geq d(\bC)  \geq  \lceil\sqrt{n}\, \rceil.
$$
Consequently,
$$
\Plotkin(\bC, \bC^\perp) = \min\{2d(\bC), d(\bC)^\perp \} \geq \lceil\sqrt{n}\, \rceil.
$$
It follows from Theorem \ref{thm-evanlint} that $\bC'$ is permutation-equivalent to $\Plotkin(\bC, \bC^\perp)$.
By Theorem \ref{thm-3.1}, $\bC'$ is self-dual. The desired conclusions then follow.
\end{proof}

Since there are infinitely many primes $n$ of the form $n \equiv 7 \pmod{8}$, these codes in Theorem
\ref{thm-Oct16}
form an infinite family of self-dual binary cyclic codes with a square-root-like
lower bound. The first seven codes in this family have the following parameters:
$$
[14,7,4], \ [46, 23, 8], \ [62, 31, 8], \ [94,47,12], \ [142,71,12], \ [158,79,16], \ [206, 103, 20].
$$

\section{Some other known self-dual linear codes with a square-root-like lower bound}\label{sec-othercodes}

There are a huge number of references about self-dual linear codes (see, e.g.,
\cite{Rains,RSsurvey} and the references therein). Since it is not easy to construct
infinite families of self-dual linear codes with a square-root-like lower bound,
we would mention some known families of such codes which were not reported in
Section \ref{sec-motivation}. Below is a list of them that the authors are aware of.
\begin{itemize}
\item Some of the families of extended binary duadic codes in \cite{LLD,LLQ,TangDing}.
\item Several families of negacyclic codes in \cite{XieChenDingSun}.
\end{itemize}
These infinite families of self-dual linear codes above are not cyclic.

\section{\begin{color}{blue}Conclusions and remarks\end{color}}\label{sec-final}

The contributions of this paper are summarized as follows.

\begin{itemize}
\item The original van Lint construction of  binary cyclic codes was generalized into a
         construction of cyclic codes over $\bF_{2^s}$ (see Theorem \ref{thm-evanlint}). In
         addition, a condition in  the original van Lint construction of  binary cyclic codes
         was removed.
\item An infinite family of Euclidean self-dual cyclic codes over $\bF_q$ for even $q$ with
        a square-root-like lower bound on their minimum distances was constructed (see
        Theorem \ref{thm-5.1}).

        When $\mu=1$ and $q\geq 4$, the cyclic code $\bC'$ in Theorem \ref{thm-5.1} has parameters
$$
[2(q^m-1), q^m-1, d \geq q^{(m+1)/2}-q+1].
$$
In this case, the lower bound $d \geq q^{(m+1)/2}-q+1$ is better than the square-root bound.

         Hence, those codes $\bC'$ in Theorem \ref{thm-5.1} defined by $\mu=1$ and $q\geq 4$ form the first infinite family of Euclidean self-dual cyclic codes whose minimum distances
         are better than the square-root bound in the literature. This makes a breakthrough
         in the construction of self-dual cyclic codes with their minimum distances better than
         the square-root lower bound.

 When $\mu=1$ and $q=2$, the cyclic code $\bC'$ in Theorem \ref{thm-5.1} has parameters
$$
[2(2^m-1), 2^m-1, d \geq 2^{(m+1)/2}].
$$
In this case, the lower bound $d \geq 2^{(m+1)/2}$ is exactly the square-root bound.
This is the first infinite family of self-dual binary cyclic codes with the square-root lower bound
and makes another breakthrough in this direction.

\item  An infinite family of Hermitian self-dual cyclic codes over $\bF_q$ for even $q$ with
        a square-root-like lower bound on their minimum distances was constructed (see
        Theorem \ref{thm-5.2}).

\item A variant of the  $[{\bf u}|{\bf u}+{\bf v}]$ construction of Euclidean self-dual linear codes was
         presented in Theorem \ref{thm-6.1}. With this new construction method, an infinite family of Euclidean self-dual linear codes over $\bF_q$ for $q \equiv 1 \pmod{4}$ with
        a square-root-like lower bound on their minimum distances was constructed in
        Theorem \ref{thm-6.2}.

        When $\mu=1$, the linear code in Theorem \ref{thm-6.2} has parameters
$$
[2(q^{m}-1), q^m-1, d \geq q^{(m+1)/2}-q+1].
$$
In this case, the lower bound $ d \geq q^{(m+1)/2}-q+1$ is better than the square-root bound as $q \geq 5$ by
definition.
\item An infinite family of self-dual binary cyclic codes with parameters $[2n, n, d \geq \lceil\sqrt{n}\, \rceil ]$ was
constructed (see Theorem \ref{thm-Oct16}).
\end{itemize}
The two breakthroughs mentioned above are the main contributions of this paper. \\

\begin{color}{blue} 
The results of this paper showed that there are an infinite family of self-dual binary cyclic codes with the 
square-root lower bound on their minimum distances and an infinite family of self-dual nonbinary cyclic codes with a lower bound 
on their minimum distances better than the square-root lower bound.  Then the following open problem  
becomes interesting. 

\begin{problem} 
Is there an infinite family of self-dual binary cyclic codes with a lower bound 
on their minimum distances better than the square-root lower bound? 
\end{problem}

Theorems \ref{thm-3.1} and \ref{thm-6.1} are very useful tools for constructing self-dual codes.  
To construct a self-dual code with a square-root-like lower bound with Theorem  \ref{thm-3.1} 
or \ref{thm-6.1},  one needs to select or design a proper 
building block, i.e., the underlying dual-containing code $\bC$.  This is the key to success. \\  
\end{color}

In this paper, we did not consider the even $m$ case, as the corresponding Euclidean and Hermitian self-dual
codes are relatively weaker in terms of their minimum distances and lower bounds on their minimum distances,
compared with the codes in the odd $m$ case presented in this paper.

\begin{color}{blue} 
\section*{Acknowledgements} 

The authors are grateful to the reviewers and the associate editor, Dr. Eimear Byrne, for their comments and suggestions that improved the presentation of this paper. 
\end{color}


\begin{thebibliography}{10}

%\bibitem{Gaborit3} C. Aguilar Melchor and P. Gaborit, On the clssification of extremal $[36, 18, 8]$  binary self-dual codes, IEEE Trans. Inf. Theory, vol. 54, no. 10, pp. 4743-4750, 2008.

\bibitem{Aly} S. Aly, A. Klappencker and P. K. Sarvepalli, On quantum and classical BCH codes, IEEE Trans. Inf. Theory, vol. 53, no. 3, pp. 1183-1188, 2007.

\bibitem{AK98}
E. F. Assmus Jr., J. D. Key,  Polynomial codes and finite geometries, In: Pless V.S., Huffman W.C. (eds.) Handbook of Coding Theory, vol. II,  pp. 1269--1343. Elsevier, Amsterdam, 1998.


%\bibitem{Bassa} A. Bassa and H. Stichtenoth, Self-dual codes better than the Gilbert-Varshamov bound, Des., Codes and Cryptogr., vol. 87, pp. 173-182, 2019.




\bibitem{BC1} R. C. Bose and D. K. Ray-Chaudhuri, On a class of error-correcting binary group codes, Inform. and Control, vol. 3, no. 1, pp. 68-79, 1960.

\bibitem{BC2} R. C. Bose and D. K. Ray-Chaudhuri, Further results on error-correcting binary group codes, Inform. and Control, vol. 3, no. 3, pp. 279-290, 1960.

%\bibitem{CRHT} X. Chen, I. Reed, T. Helleseth and T. K. Troung, Use of Grobner bases to decode binary cyclic codes up to the true minimum ditances, IEEE Trans. Inf. Theory vol. 40, no. 5, pp. 1654-1661, 1994.

%\bibitem{ChenTing} T. Chen, C. Ding, C. Li and Z. Sun, Four infinite families of ternary cyclic codes with a square-root-like lower bound, arXiv:2303.06849, 2023.



%\bibitem{Conway1968} J. H. Conway, A perfect group of order 8315553613086720000 and the sporadic simple groups, Proc. Nat. Acad. Sci., vol. 61, 398-400, 1968.


%\bibitem{CPS} J. H. Conway, V. Pless and N. J. A. Sloane, Self-dual codes over $GF(3)$ and $GF(4)$ of length not exceeding $16$, IEEE Trans. Inf. Theory, vol. 25, no. 3, pp. 312-322, 1979.

%\bibitem{CS} J. H. Conway and N. J. A. Sloane, A new upper bound on the minimal distance of self-dual codes, IEEE Trans. Inf. Theory vol. 36, pp. 1319-1333, 1990.

\bibitem{dingtang2022}
C. Ding, C. Tang, Designs from Linear Codes, Second Edition, World Scientific, Singapore, 2022.



%\bibitem{DGH} S. T. Dougherty, T. A. Gullier and M. Harada, Extremal binary self-dual codes, IEEE Trans. Inf. Theory vol. 43, no. 6, pp. 2036-2047, 1997.




\bibitem{Gaborit2} P. Gaborit, Tables of self-dual codes, Tables de codes auto-duaux, http://www.unilim.fr/pages-perso/phillie.gaborit.





\bibitem{Gaborit} P. Gaborit and A. Otmani, Experimental constructions of self-dual codes, Finite Fields Appl., vol. 9, pp. 372-394, 2003.


%\bibitem{Golay} M. J. E. Golay, Notes on digital coding, Proc. IEEE, vol. 37, 657, 1949.

%\bibitem{GG08} M. Grassl and T. A. Gulliver, On self-dual MDS codes, Proc. Int. Symp. Inf. Theory, pp. 1954-1957, 2008.



%\bibitem{Gulliver2} T. A. Gulliver, Optimal double circulant self-daul codes over ${\bf F}_4$, IEEE Trans. Inf. Theory, vol. 46, no. 9, pp. 271-274, 2000.



\bibitem{HTop} B. Heijne and J. Top, On the minimum distance of binary self-dual cyclic codes, IEEE Trans. Inf. Theory, vol. 55, no. 11, pp. 4860-4863, 2009.

\bibitem{Hoc} A. Hocquenghem, Codes correcteurs d'erreurs, Chiffres (Paris), vol. 2, pp. 147-156, 1959.



\bibitem{HP} W. C. Huffman and V. Pless, Fundamentals of error-correcting codes, Cambridge University Press, Cambridge, U. K., 2003.

%\bibitem{Kai} X. Kai and S. Zhu, On cyclic self-dual codes, Appl. Algebra Engr.Commun. Comput., vol. 19, pp. 509-525, 2008.


%\bibitem{KS} A. Krishna and V. Sarwate, Pseudocyclic maximal-distance-separable codes, IEEE Trans. Inf. Theory, vol. 36, no. 4, pp. 880-884, 1990.

%\bibitem{Jia} Y. Jia, S. Ling and C. Xing, On self-dual cyclic codes over finite fields, IEEE Trans. Inf. Theory, vol. 57, no. 4, pp. 2243-2251, 2011.




%\bibitem{Leech} J. Leech, Notes on sphere packings, Canad. J. Math., vol. 19, 251-267, 1967.

%\bibitem{Leon} J. S. Leon, V. Pless and N. J. A. Sloane, Duadic codes, IEEE Trans. Inf. Theory, vol. 30, pp. 709-714, 1981.


\bibitem{LDL} C. Li, C. Ding and S. Li, LCD cyclic codes over finite fields, IEEE Trans. Inf. Theory, vol. 63, no. 7, pp. 4344-4356, 2017.

\bibitem{LLD}
H. Liu, C. Li, C. Ding, Five infinite families of binary cyclic codes and their related codes with good
parameters,  Finite Field Appl., vol. 91, 102270,  2023.

\bibitem{LLQ}
H.  Liu, C. Li, H. Qian, Parameters of several families of binary duadic codes and their related codes,
Des. Codes Crypt.,  \begin{color}{blue} vol. 92, no. 1, pp. 1--12, Jan. 2024. \end{color} 


\bibitem{MScode} F. J.  MacWilliams and N. J. A. Sloane, The theory of error-correcting codes, 3rd Edition, North-Holland Mathematical Library, vol. 16. North-Holland, Amsterdam, 1977.





\bibitem{Rains} G. Nebe, E. M. Rains and N. J. A. Sloane, Self-dual codes and invariant theory, Algorithms and Computation in Mathematics, vol. 17, eds, A. M. Cohen, H, Cohen, D. Eisenbud and B. Sturmfels, Springer-Verlag Berlin Heidelbeg, 2006.

\bibitem{Pless1972}
V. Pless, Symmetry codes over GF $(3)$ and new $5$-designs, \emph{J. Comb. Theory}, vol. 12, pp. 119--142, 1972.


\bibitem{Plotkin}
M. Plotkin, Binary codes with specified minimum distance, IRE Trans., vol. IT-6, pp. 445--450, 1960.

\bibitem{RSsurvey}
E. M. Rains and N. J. A. Sloane, Self-dual codes,  In: Pless V.S., Huffman W.C. (eds.) Handbook of Coding Theory, vol. I,  pp. 174--299. Elsevier, Amsterdam, 1998.

\bibitem{Sloane} N. J. A. Sloane and J. G. Thompson, Cyclic self-dual codes, IEEE Trans. Inf. Theory, vol. 29, no. 3, pp. 364-366, 1983.



\bibitem{SunDW} Z. Sun, C. Ding and X. Wang, Two classes of constacyclic does with variable parameters
$[(q^m-1)/r, k, d]$, 
\begin{color}{blue}IEEE Trans. Inf. Theory,  vol. 70, no. 1, pp. 93--114, Jan.  2024. \end{color} 


%\bibitem{SunLiDing} Z. Sun, C.Li and C. Ding, An infinite family of binary cyclic codes with best parameters, to appear in IEEE Trans. Inf. Theory, 2023.

\bibitem{TangDing} C. Tang and C. Ding, Binary $[n, \frac{n+1}{2}]$ cyclic codes with good minimum distances, IEEE Trans. Inf. Theory, vol. 68, no. 12, pp. 7842-7849, 2022.


\bibitem{Lint} J. H. van Lint, Introduction to the coding theory, GTM 86, Third and Expanded Edition, Springer, Berlin, 1999.

\bibitem{Lint1} J. H. van Lint, Repeated-root cyclic codes, IEEE Trans. Inf. Theory, vol. 37, no. 2, pp. 343-345, 1991.



\bibitem{XieChenDingSun} C. Xie, H. Chen, C. Ding and Z. Sun, Self-dual negacyclic codes with variable lengths and square-root-like lower bounds on the minimum distances,  
\begin{color}{blue} 
 IEEE Trans. Inf. Theory,  vol. 70, no. 7,  pp. 4879--4888, July 2024. 
\end{color} 





\end{thebibliography}
\end{document}